\documentclass[american,aps,pra,reprint, superscriptaddress]{revtex4-1}
\usepackage{amsmath,amscd,amsthm,amssymb}
\usepackage{mathrsfs}
\usepackage{graphicx}
\usepackage{epsf,epstopdf}
\usepackage{caption3}
\usepackage[all]{xy}
\usepackage{tikz}
\usepackage{enumerate}

\usepackage[unicode=true,pdfusetitle, bookmarks=true,bookmarksnumbered=false,bookmarksopen=false, breaklinks=false,pdfborder={0 0 0},backref=false,colorlinks=false] {hyperref}
\hypersetup{ colorlinks,linkcolor=myurlcolor,citecolor=myurlcolor,urlcolor=myurlcolor}
\usepackage{graphics, times, braket, color, bm}
\usepackage[up]{subfigure}
\usepackage{cleveref}

\definecolor{myurlcolor}{rgb}{0,0,0.7}

\def\be{\begin{equation}}
\def\ee{\end{equation}}
\def\bea{\begin{eqnarray*}}
\def\eea{\end{eqnarray*}}
\def\ot{\otimes}

\theoremstyle{plain}
\newtheorem{thrm}{\protect\theoremname}
\newtheorem{defi}[thrm]{Definition}
\newtheorem{prop}[thrm]{Proposition}

\providecommand{\theoremname}{Theorem}

\newcommand{\iinner}[2]{\langle #1 | #2\rangle}
\newcommand{\out}[2]{| #1\rangle\langle #2 |}
\DeclareMathOperator{\trace}{tr}
\newcommand{\ptr}[2]{\trace_{#1}({#2})}
\newcommand{\tr}[1]{\ptr{}{#1}}
\newcommand{\id}{\mathbf{id}}
\newcommand{\ids}[1]{\id_{#1}}
\newcommand*{\myproofname}{Proof}

\makeatother

\def\cC{\mathcal{C}}\def\cD{\mathcal{D}}
\def\cH{\mathcal{H}}\def\cI{\mathcal{I}}
\def\cO{\mathcal{O}}
\def\cS{\mathcal{S}}

\def\bE{\mathbf{E}}

\def\bP{\mathbf{P}}

\def\rC{\mathrm{C}}

\def\mCMI{\textmd{CMI}} \def\mI{\textmd{I}} \def\mPI{\textmd{PI}}

\def\fC{\mathfrak{C}}\def\fP{\mathfrak{P}}

\theoremstyle{definition}

\theoremstyle{remark}

\graphicspath{{figs/}}

\begin{document}

 \author{Sunho Kim}
 \email{kimsunho81@hrbeu.edu.cn}
 \affiliation{School of Mathematical Sciences, Harbin Engineering University, Harbin 150001, China}

  \author{Chunhe Xiong}
 \email{xiongchunhe@zju.edu.cn}
 \affiliation{School of Computer and Computing Science, Zhejiang University City College, Hangzhou 310015, China}

\author{Asutosh Kumar}
 \email{asutoshk.phys@gmail.com}
 \affiliation{P.G. Department of Physics, Gaya College, Magadh University, Rampur, Gaya 823001, India}
 \affiliation{Harish-Chandra Research Institute, HBNI, Chhatnag Road, Jhunsi, Allahabad 211019, India}
 \affiliation{Vaidic and Modern Physics Research Centre, Bhagal Bhim, Bhinmal, Jalore 343029, India}

\author{Guijun Zhang}
 \affiliation{School of Mathematical Sciences, Zhejiang University, Hangzhou 310027, China}

 \author{Junde Wu}
 \email{wjd@zju.edu.cn}
 \affiliation{School of Mathematical Sciences, Zhejiang University, Hangzhou 310027, China}

\title{Quantifying dynamical coherence with coherence measures}
\begin{abstract}
Quantum coherence, like entanglement, is a fundamental resource in quantum information.
In recent years, remarkable progress has been made in formulating resource theory of coherence from a broader perspective. The notions of block-coherence and POVM-based coherence have been established.
Certain challenges, however, remain to be addressed. It is difficult to define incoherent operations directly, without requiring incoherent states, which proves a major obstacle in establishing the resource theory of  dynamical coherence.
In this paper, we overcome this limitation by introducing an alternate definition of incoherent operations, induced via coherence measures, and quantify dynamical coherence based on this definition. Finally, we apply our proposed definition to quantify POVM-based dynamical coherence. 
\end{abstract}
\maketitle

\section{Introduction}

We are quite familiar with the macro-world. 
In the micro-world, however, a given physical system may exhibit remarkable properties that are not seen in the macro-world. These properties of quantum systems offer certain advantages over classical systems in information processing tasks. Consequently, these properties may be viewed as ``resources'' in those tasks. 
We consider a certain physical situation and identify the fundamental or practical constraints associated with it: what are the tools available, what things are free of cost (these can be spent amply), and what involves a cost (these need to be used wisely). 
In quantum physics, things we consider are states and operations. In a given situation, a particular set of states and operations are free (in the sense that states can be readily prepared and operations are closed or convex) while another set of them is viewed as resources because they cannot be obtained with free states and free operations. Put differently, these objects cannot be generated in the given setting without incurring a cost.
A quantum resource theory identifies physical processes and quantum states as being either free, restricted or resources. This is a common feature of any resource theory.

Several properties of quantum sysetms have been recognized as resources in quantum information.
This recognition is profound as they can be unified and developed under the same
roof of quantum resource theories (QRTs) \cite{Chitambar, Brandao}.
Such a unification leads to better understanding of physical phenomena and new discoveries.
This helps in identifying structures and applications that are common to resource theories in general. Furthermore, QRTs study what information processing tasks are possible using the restricted operations.
We should note that the structure of resource theories goes far beyond quantum physics.
A basic goal of any resource theory is to quantify the resource and, with several resources at hand, find necessary and sufficient conditions that determine whether or not one resource can be converted to another by the set of free operations.
Quantum information theory is regarded a theory of interconversions among different resources  that are classified from diverse perspectives as classical or quantum, noisy or noiseless, and
static (i.e., quantum states) or dynamic (i.e., quantum channels).
The resource theory framework is versatile and has
been adopted in several areas of quantum information and physics.
QRTs provide a structured framework that quantitatively describes quantum properties such as quantum correlation \cite{Horodecki, Ollivier}, quantum coherence \cite{Streltsov1, Streltsov2}, quantum reference frames \cite{Bartlett, Gour}, and nonlocality \cite{Bell, Brunner}, etc.
Although there is a large degree of freedom in defining the free states and free operations, striking similarities emerge among different QRTs in terms of resource measures and resource convertibility. Consequently, various properties that appear distinct on the surface possess great similarity on a deeper structural level.
For example, quantum coherence in multipartite systems embodies the essence of entanglement and is a key ingredient for a plurality of physical phenomena in quantum physics and quantum information.

Quantum coherence marks the departure of quantum physics from classical physics.
Unlike entanglement, coherence can be present in both single and composite quantum systems.
Several studies have been devoted in recent years to characterize and quantify the coherence in a quantum physical system \cite{Baumgratz, Girolami, Lostaglio, Shao, Pires, Rana, Rastegin, Napoli, Yu, Luo, Bu, Xiong}. Measures of coherence have interesting operational interpretations.
%
Coherence as a basic resource is very promising. It has been found to offer significant quantum advantages in various quantum information processing tasks such as quantum biology \cite{Huelga}, quantum thermodynamics \cite{Mitchison, Kammerlander, Korzekwa}, quantum algorithms \cite{Hillery}, and steering \cite{Mondal}.

Mathematically, quantum coherence is represented as the presence of the off-diagonal terms of a density matrix in a preferred basis.
In this formalism, it is a static (state-based) resource. While static coherence is the degree of superposition present in a state, dynamic coherence is the ability to generate, preserve or distribute coherence. Dynamic coherence essentially studies coherence from the viewpoint of quantum operations.
Recent studies have shown that dynamical coherence, like static coherence \cite{Streltsov4, Ma, Tan}, can be measured by different dynamical quantum resources \cite{Theurer}, and provide advantages in various quantum tasks \cite{Ducuara, Uola}.

Since the incoherent (free) states are represented by diagonal matrices in a fixed orthonormal basis, they can be considered associated with a von Neumann measurement. From this connection of coherence and quantum measurement, the resource theory of coherence has been expanded to coherence based on general projective measurements such as block-coherence \cite{Aberg} and coherence based on positive-operator-valued measures (POVMs) \cite{Bischof1, Bischof2}, which are connected with each other through Naimark's theorem \cite{Paulsen, Watrous}.
The Naimark extension, i.e., embedding the states and operations into a higher-dimensional space allows for a simpler derivation of generic results, and the possibility of direct implemention of any POVM in an experiment.
POVMs describe the most general quantum measurements and can provide operational advantages compared with any projective measurement \cite{Oszmaniec}.
Because of the rich structure of POVMs, the resource theoretic formulation of POVM-based coherence generalizes and reveals features that are distinct from the standard resource theory of coherence \cite{Bischof1, Bischof2}.
Moreover, an important operational interpretation of POVM coherence in terms of {\it cryptographic randomness gain} is given in \cite{Bischof2}.
In Ref. \cite{Kim-povm} authors provided a connection of POVM-coherence with entanglement. In particular, they showed that it is possible to convert POVM-coherence into entanglement and suggested some practical strategies.

A resource theory of dynamical quantum coherence has been formulated recently by identifying classical channels as the free elements and considering the preservation of coherence a resource \cite{Saxena}. They introduced four different types of free ``superchannels'', quantified dynamical coherence using channel-divergence-based monotones, and provided operational interpretation to these monotones. 
As a comparing note to this, we propose that our definition can also be extended to redefine free superchannels between different dynamical systems. In present work, however, we focus on complementing the difficulty of defining POVM-based incoherent operations and quantifying dynamical resources.

The traditional framework based on the premise of a convex set of incoherent states, however, has certain limitations on quantifying the POVM-based dynamical coherence in situations where the existence of incoherent states is not guaranteed.
In this paper, we propose a new framework that overcomes limitations on quantifying POVM-based dynamical coherence.

The paper is organized as follows. Section \uppercase\expandafter{\romannumeral2} introduces POVM-based coherence. In Sec. \uppercase\expandafter{\romannumeral3} we define incoherent operations induced by coherence measures.
The dynamical coherence is quantified in the newly defined framework in Sec. \uppercase\expandafter{\romannumeral4}, and we prove that the proposed measure is faithful. In Sec. \uppercase\expandafter{\romannumeral5}, we apply our new framework for quantifying dynamical coherence.
Section \uppercase\expandafter{\romannumeral6} concludes with a summary.

\section{Resource theory of quantum coherence}

Since the establishment of the resource theory of quantum coherence by Baumgratz {\it et al.}, various interpretations and studies have been conducted. Moreover, the notion of block coherence associated with a projective measurement in which the rank of at least one measurement operator exceeds unity was established in a framework similar to the typical coherence \cite{Baumgratz, Aberg, Bischof1, Bischof2}. 
Block coherence is defined with respect to a projective measurement on the set of quantum states. See Refs. \cite{Aberg, Bischof1} for details on block coherence and associated measures. On the contrary, coherence of a quantum state defined with respect to an arbitrary positive-operator-valued measure (POVM) via its canonical Naimark extension is called POVM-based coherence.
In this paper, we mainly focus on POVM coherence.
Therefore, we briefly recall the resource theory of POVM-based coherence introduced by Bischof \emph{et al.} 
They considered coherence defined with respect to arbitrary POVMs, not limited to projective measurements.
The generalisation requires the following theorem:
\begin{itemize}
  \item (\emph{Naimark Theorem}) A POVM $\bE = \{E_i\}_{i=0}^{n-1}$ on $\cH^{S_0}$ with $n$ outcomes can be extended to a projective measurement $\bP = \{P_i\}_{i=0}^{n-1}$ on the Naimark space $\cH^{S}$ (with dimension $d_{S}\geq d_{S_0}$) such that
\bea
\tr{E_i\rho} = \textmd{tr}\big\{P_i(\rho\ot\out{0}{0})\big\}
\eea
holds for all states $\rho$ in $\cS_0$.
\end{itemize}

Using this theorem, the following two definitions are proposed:
\begin{enumerate}[D1]
        \item (\emph{POVM-based coherence measure}) A POVM-based coherence measure $C_{\bE}(\rho^{S_0})$ for a state $\rho^{S_0}$ in $\cS_0$ is defined in terms of the block coherence of the embedded state $\Phi[\rho^{S_0}] = \rho^{S_0}\ot \out{0}{0}^{S_1}$ with respect to a canonical Naimark extension $\bP$ of the POVM $\bE$ as,
\be\label{eq:1}
C_{\bE}(\rho^{S_0}) := C_{\bP}(\Phi[\rho^{S_0}]),
\ee
where $C_{\bP}(\rho^{S})$ is a unitarily invariant block-coherence measure on $\cS$--the set of quantum states on the Naimark space $\cH^S = \cH^{S_0}\ot \cH^{S_1}$.

     \item\label{df:2} (\emph{POVM-based incoherent operations}) A (maximally) POVM-based incoherent operation is defined as $\Lambda_{\textmd{MPI}} := \Phi^{-1}\circ \Lambda\circ \Phi$, where $\Lambda$ is a block incoherent operation  with respect to $\bP$ that satisfies $\Lambda[\cS_{\Phi}] \subseteq \cS_{\Phi}$
 for the subset $\cS_{\Phi} \subseteq \cS$ of embedded system states $\Phi[\rho^{S_0}]$. Let $\cO_{\mPI}^{\bE}$ denotes the set of POVM($\bE$)-based incoherent operations.
\end{enumerate}

The POVM-based coherence measure is well defined by the underlying block-coherence measure in Eq.(\ref{eq:1}), and has the following properties:
\begin{enumerate}[P1]
        \item (\emph{Faithfulness}): $C_{\bE}(\sigma^{S_0}) \geq 0$, and $C_{\bE}(\sigma^{S_0}) = 0$ if and only if
\be\label{eq:2}
\sum_i\overline{E}_i\sigma^{S_0}\overline{E}_i = \sigma^{S_0},
\ee
where $\overline{E}_i$ denotes the projective part of $E_i$, \emph{i.e.}, the projector onto the range of $E_i$.
     \item (\emph{Monotonicity}):  $C_{\bE}(\Lambda_{\textmd{MPI}}[\rho])\leq C_{\bE}(\rho)$ for all POVM-based incoherent
operations $\Lambda_{\textmd{MPI}}$ of the POVM $\bE$.
     \item (\emph{Convexity}): $C_{\bE}(\sum_i p_i \rho_i^{S_0}) \leq \sum_i p_i C_{\bE}(\rho_i^{S_0})$ for any set of quantum states $\{\rho_i^{S_0}\}$ and probability distribution $(p_i)$.
\end{enumerate}

We should note here that POVM-based coherence measures are determined via the block coherence measures selected in the Naimark space and the definition of POVM-based coherence.
These POVM-based coherence measures are, however, independent of the choice of the Naimark space \cite{Bischof1}.
Then, the relative entropy of POVM-based coherence measure $C_{\bE,r}(\rho^{S_0})$ is derived from the relative entropy of block-coherence measure by Eq.(\ref{eq:1}). It can be expressed in the form:\\
\parbox{8.2cm}{
\begin{eqnarray*}
  C_{\bE,r}(\rho^{S_0})&=& C_{\bP,r}(\Phi[\rho^{S_0}]) = \min_{\sigma\in\cI_{\textmd{BI}}}S(\Phi[\rho^{S_0}]\parallel\sigma)\\
  &=& H[\{p_i\}] + \sum_ip_iS(\rho_i) - S(\rho^{S_0}),
\end{eqnarray*}}\hfill
\parbox{.3cm}{\begin{eqnarray}\label{eq:3}\end{eqnarray}}\\
where $p_i = \tr{E_i\rho^{S_0}},\ \rho_i$ is the $i$-th postmeasurement state for a given measurement operator $A_i$, \emph{i.e.}, $\rho_i = (1/p_i)A_i\rho^{S_0}A_i^{\dagger}$, and the Shannon entropy $H[\{p_i\}] = -\sum_ip_i\log p_i$.

Bischof \emph{et al.} have shown in Ref. \cite{Bischof1} that these POVM-based coherence measure and POVM-based incoherent operations are independent of the choice of Naimark extension. In particular, the set of POVM-based incoherent operations can be characterized by a semidefinite feasibility problem (SDP).
Moreover, it is still uncertain if POVM-based incoherent operations can be characterized without Naimark extension; POVM-based incoherent operations are realized only by the block incoherent operations defined in D\ref{df:2}.
Unlike general coherence theory, free states do not materialize or exist in the POVM-based coherence. Also, it is either limited or impossible to define free operations, using general resource theory, within a single system.
Below we introduce an alternate definition of incoherent operations to overcome the limitations mentioned above. These can be applied more conveniently and extensively.

\section{Incoherence operations induced by convex measure of coherence}

Our aim here is to establish a new framework for the dynamical coherence analogous to the static framework.
We denote by $\cD_{S}$ the set of quantum states on system $S$.
Next we define the incoherent operations induced by a valid convex measure of coherence.
In the theory of coherence established by Baumgratz {\it et al.}, incoherent states are required to define incoherent operations. Here we propose an alternate definition of incoherent operations that does not require incoherent states necessarily.

\begin{defi}\label{df:1}
We define a completely-positive and trace-preserving (CPTP) map $\Lambda$ an incoherent operation induced by a convex coherence measure $C$ (CMIO by $C$) on $S$ if the map $\Lambda$ satisfies $C\big(\Lambda[\rho]\big)\leq \rC(\rho)$ for all $\rho\in \cD_S$. We denote by $\cO_{\mCMI}^{C}$ the set of incoherent operations induced by $C$.
\end{defi}

According to this definition, a given operation is considered incoherent if it does not increase resources for all quantum states.
Moreover, we can see that $\cO_{\mCMI}^{C}$ is a nonempty convex set which necessarily includes the identity operation $\ids{S}$, induced by convex measures, such that $C\big(\Lambda[\rho]\big)\leq \rC(\rho)$.
Also, the definition of incoherent operation induced by measures in the resource theory of general coherence (in which the existence of incoherent states is assumed) is consistent with that of the existing incoherent operation. It can be confirmed by the following proposition.

\begin{prop}\label{pr:1}
In the resource theory of general coherence, a CPTP map $\Lambda$ is a CMIO by $C$ in Def. \ref{df:1} if and only if $\Lambda$ is a general incoherent operation. That is, $\cO_{\mCMI}^{C} = \cO_{\mI}$, where $\cO_{\mI}$ is the set of incoherent operations based on the existing definition.
\end{prop}
\begin{proof}
If $\Lambda\in\cO_{\mCMI}^{C}$ then $C\big(\Lambda[\rho]\big) = C(\rho) = 0$ for any $\rho\in \cI$, i.e., $\Lambda[\cI]\subset \cI$. This implies $\Lambda\in\cO_{\mI}$. On the other hand, if $\Lambda\in\cO_{\mI}$ then from the monotonicity of the convex coherence measure $C$, we have $C\big(\Lambda[\rho]\big) \leq C(\rho)$ for all $\rho$. Thus $\Lambda\in\cO_{\mCMI}^{C}$ by Def. \ref{df:1}.
\end{proof}

Prop. \ref{pr:1} shows that Def. \ref{df:1} is equally applicable in the resource theory of general coherence. Furthermore, we can also define POVM($\bE$)-based incoherent operations on the original single system without considering the block incoherent operations on the Naimark system and the POVM-based incoherent states. A CPTP map $\Lambda^{S_0}$ on the system $S_0$ is a CMIO by $C_{\bE}$, if it satisfies
\be\label{eq:4}
\begin{array}{ccc}
  C_{\bE}(\Lambda^{S_0}[\rho^{S_0}]) & \leq &  C_{\bE}(\rho^{S_0})\\
    & &  \\
  \parallel & &  \parallel\\
    & &  \\
  C_{\bP}(\Phi[\Lambda^{S_0}[\rho^{S_0}]])\quad & \leq & \quad C_{\bP}(\Phi[\rho^{S_0}])
  \end{array}
\ee
for any quantum state $\rho^{S_0}$, where $\bP$ is a Naimark extension of the POVM on $S_0$. This is clearly a simpler and more convenient form than the definition of POVM-based incoherent operations in D\ref{df:2}.
It should be, however, noted here that our new definition applies only to a subset of $\cS$. In other words, the POVM-based CMIOs by $C_{\bE}$ are obtained by applying Def. \ref{df:1} only to the subset $\cS_{\Phi}(=\{\Phi[\rho^{S_0}]|\rho^{S_0} \in \cS_{0}\})$ of $\cS$ which implies that the POVM-based CMIO by $C_{\bE}$ is a weaker notion than the POVM-based incoherent operation, \emph{i.e.}, $\cO_{\mPI}^{\bE} \subseteq \cO_{\mCMI}^{C_{\bE}}$ (from Prop. \ref{pr:1}, the two definitions are consistent when applied for the whole set). However, if we can find block incoherent operations satisfying $\Lambda_{\textmd{MPI}} = \Phi^{-1}\circ \Lambda\circ \Phi$ for all POVM-based CMIOs by $C_{\bE}$, we can prove $\cO_{\mPI}^{\bE} = \cO_{\mCMI}^{C_{\bE}}$, which would mean that the POVM-based CMIOs are independent of the choice of $C_{\bE}$. Here we leave this as an open problem whether or not $\cO_{\mPI}^{\bE} = \cO_{\mCMI}^{C_{\bE}}$.

Since we have already established a new definition of incoherent operations, we can consider erecting a theory to quantify dynamic resources from Def. \ref{df:1}.
In the following section, we propose a protocol for quantifying dynamical resources of incoherent operations induced by measures in the POVM-based coherence theory.

\section{Quantifying dynamical coherence induced by POVM-based convex coherence measures}

Let $\bE$ be a POVM on $\cH^{S}$ and $C_\bE$ a convex measure of coherence associated with POVM $\bE$. To quantify dynamical resources based on the measure $C_\bE$, we define the following relative function $\fP_{C_\bE}(\Theta;\Lambda)$ on CPTP maps $\Theta$ and $\Lambda$ on $\cH^{S}$:
\be\label{eq:5}
\fP_{C_\bE}(\Theta;\Lambda) := \sup_{\rho\in\cD_S}\big\{C_\bE\big(\Theta[\rho]\big)-C_\bE\big(\Lambda[\rho]\big)\big\}.
\ee
Since $\cD_S$ is a closed set, the supremum over the states in this definition is in effect the maximum, \emph{i.e.},
\bea
\fP_{C_\bE}(\Theta;\Lambda) = \max_{\rho\in\cD_S}\big\{C_\bE\big(\Theta[\rho]\big)-C_\bE\big(\Lambda[\rho]\big)\big\}.
\eea
Thus, $\fP_{C_\bE}(\Theta;\Lambda)$ is the maximal relative increase in the resource that the map $\Theta$ can extract compared with the map $\Lambda$ for all quantum states. We call the function $\fP_{C_\bE}(\Theta;\Lambda)$ the power of $\Theta$ over $\Lambda$. It should be mentioned here that we focus only on relative increases and not the relative differences in resources through the two maps, and therefore the function $\fP_{C_\bE}$ can assume negative values.
Note that $\fP_{C_\bE}$ is not symmetric in its arguments $\Theta$ and $\Lambda$, \emph{i.e.,} $\fP_{C_\bE}(\Theta;\Lambda)\neq \fP_{C_\bE}(\Lambda;\Theta)$.
Moreover, $\fP_{C_\bE}$ satisfies the convexity:\\
\\
\parbox{6.2cm}{
\begin{eqnarray*}
&&\fP_{C_\bE}(p\Theta_1+(1-p)\Theta_2;\Lambda)\\
&&\quad =  \max_{\rho\in\cD_S}\big\{C_\bE\big(p\Theta_1[\rho]+(1-p)\Theta_2[\rho]\big)-C_\bE\big(\Lambda[\rho]\big)\big\}\\
&&\quad \leq \max_{\rho\in\cD_S}\big\{pC_\bE\big(\Theta_1[\rho]\big)+(1-p)C_\bE\big(\Theta_2[\rho]\big)-C_\bE\big(\Lambda[\rho]\big)\big\}\\
&&\quad \leq p\max_{\rho_1\in\cD_S}\big\{C_\bE\big(\Theta_1[\rho_1]\big)-C_\bE\big(\Lambda[\rho_1]\big)\big\}\\
&&\qquad +\ (1-p)\max_{\rho_2\in\cD_S}\big\{C_\bE\big(\Theta_2[\rho_2]\big)-C_\bE\big(\Lambda[\rho_2]\big)\big\}\\
&&\quad =p\fP_{C_\bE}(\Theta_1) + (1-p)\fP_{C_\bE}(\Theta_2;\Lambda),
\end{eqnarray*}}\hfill
\parbox{.3cm}{\begin{eqnarray}\label{eq:the convexity of power}\end{eqnarray}}\\
where the first and the last equalities are by definition of $\fP_{C_\bE}$, the second inequality is due to the convexity of $C_\bE$, and the third inequality is obtained because the sum of the maximum values is not less than the maximum value for the sum.

Next, we propose a protocol that quantifies dynamical coherence induced by the measure $C_\bE$ using this power $\fP_{C_\bE}$.

\begin{defi}\label{df:3}
For a POVM $\bE$ on $\cH^{S}$ and $C_\bE$ the associated convex measure of coherence, we define
\be\label{eq:7}
\fC_{\mCMI}^{C_\bE}(\Theta) := \inf_{\Lambda\in \cO_{\mCMI}^{C_\bE}}|\fP_{C_\bE}(\Theta; \Lambda)|.
\ee
\end{defi}

This definition gives us interesting results. First, let us note that $\ids{S}\in \cO_{\mCMI}^{C_\bE}$ and $\fP_{C_\bE}(\Lambda;\ids{S})\leq0$ when $\Lambda\in\cO_{\mCMI}^{C_\bE}$. So, for $\Theta\notin\cO_{\mCMI}^{C_\bE}$ and all $\Lambda \in \cO_{\mCMI}^{C_\bE}$, we have
\bea
\fP_{C_\bE}(\Theta; \ids{S}) &\leq& \fP_{C_\bE}(\Theta; \ids{S})-\fP_{C_\bE}(\Lambda; \ids{S})\\
&\leq&  \max_{\rho\in\cD_S}\big\{C_\bE\big(\Theta[\rho]\big)- C_\bE\big(\Lambda[\rho]\big)\big\}\\
&=&\fP_{C_\bE}(\Theta; \Lambda).
 \eea

And if there is at least one quantum state $\rho$ such that $C_\bE\big(\Theta[\rho]\big) > C_\bE(\rho)$, then
  \bea
 \quad \fP_{C_\bE}(\Theta;\ids{S}) &=& \max_{\rho\in\cD_S}\{C_\bE\big(\Theta[\rho]\big)-C_\bE(\rho)\}\\
 &\geq & C_\bE\big(\Theta[\rho]\big) - C_\bE(\rho)>0.
  \eea
This means that if the CPTP map $\Theta\notin \cO_{\mCMI}^{C_\bE}$, the measure $\fC_{\mCMI}^{C_\bE}(\Theta)$ represents the maximum increase in coherence for all quantum states. In contract, if $\Theta\in \cO_{\mCMI}^{C_\bE}$, the relative increase with itself is zero. That is,
\be\label{eq:8}
\fC_{\mCMI}^{C_\bE}(\Theta) = \left\{
  \begin{array}{lr}
   \fP_{C_\bE}(\Theta; \ids{S})>0 \quad &\text{when}~~\Theta\notin \cO_{\mCMI}^{C_\bE},\\
   \fP_{C_\bE}(\Theta; \ids{S})=0 \quad &\text{when}~~\Theta\in \cO_{\mCMI}^{C_\bE}.
  \end{array}
\right.
\ee
Thus, from above results, this function allows us to determine whether a CPTP map is incoherent or not for the given coherence measure $C_\bE$. Next we prove that the measure $\fC_{\mCMI}^{C_\bE}$ justly quantifies the dynamical resource of coherence.

\begin{thrm}\label{th:1}
Let $\bE$ be a POVM on $\cH^{S}$. Then, for any convex measure $C_\bE$ of coherence based on $\bE$, $\fC_{\mCMI}^{C_\bE}(\Theta)$ is a convex measure of dynamical coherence.
\end{thrm}
\begin{proof}
\begin{enumerate}[(i)]
  \item (Faithfulness) $\fC_{\mCMI}^{C_\bE}(\Theta) \geq 0$ from Eqs. (\ref{eq:8}), and $\fC_{\mCMI}^{C_\bE}(\Lambda) = 0$ if and only if $\Lambda\in \cO_{\mCMI}^{C_\bE}$.

  \item (Convexity) We need to prove that $\quad \fC_{\mCMI}^{C_\bE}(p\Theta_1+(1-p)\Theta_2) \leq p\fC_{\mCMI}^{C_\bE}(\Theta_1)+(1-p)\fC_{\mCMI}^{C_\bE}(\Theta_2)$
  for $0\leq p \leq1$. If $p\Theta_1+(1-p)\Theta_2 \in \cO_{\mCMI}^{C_\bE}$, the above inequality is true from the faithfulness of $\fC_{\mCMI}^{C_\bE}$. If $p\Theta_1+(1-p)\Theta_2 \notin \cO_{\mCMI}^{C_\bE}$, at least one of the operations $\Theta_1$ and $\Theta_2$ is not included in $\cO_{\mCMI}^{C_\bE}$ because $\cO_{\mCMI}^{C_\bE}$ is a convex set. Let us first consider the case when $\Theta_1\notin \cO_{\mCMI}^{C_\bE}$ and $\Theta_2\notin \cO_{\mCMI}^{C_\bE}$. Then
  \bea
&&\quad \fC_{\mCMI}^{C_\bE}(p\Theta_1+(1-p)\Theta_2)\\
  &&\qquad\qquad =\fP_{C_\bE}(p\Theta_1+(1-p)\Theta_2;\ids{S})\\
  &&\qquad\qquad \leq p\fP_{C_\bE}(\Theta_1;\ids{S}) + (1-p)\fP_{C_\bE}(\Theta_2;\ids{S})\\
  &&\qquad\qquad = p\fC_{\mCMI}^{C_\bE}(\Theta_1)+(1-p)\fC_{\mCMI}^{C_\bE}(\Theta_2),
\eea
where the first and the last equalities are due to Eq.(\ref{eq:8}), and for the inequality we use the convexity (\ref{eq:the convexity of power}) of $\fP_{C_\bE}$.
Next, without loss of generality, we assume that $\Theta_1\notin \cO_{\mCMI}^{C_\bE}$ and $\Theta_2\in \cO_{\mCMI}^{C_\bE}$. Then we have the following chain of equalities and inequalities:
\bea
  &&\quad\fC_{\mCMI}^{C_\bE}(p\Theta_1+(1-p)\Theta_2)\\
  &&\qquad\qquad =\fP_{C_\bE}(p\Theta_1+(1-p)\Theta_2;\ids{S})\\
  &&\qquad\qquad \leq p\fP_{C_\bE}(\Theta_1;\ids{S}) + (1-p)\fP_{C_\bE}(\Theta_2;\ids{S})\\
  &&\qquad\qquad \leq p\fP_{C_\bE}(\Theta_1;\ids{S}) + (1-p)\fP_{C_\bE}(\Theta_2; \Theta_2)\\
  &&\qquad\qquad =p\fC_{\mCMI}^{C_\bE}(\Theta_1)+(1-p)\fC_{\mCMI}^{C_\bE}(\Theta_2).
\eea
Here again, the first and the last equalities follow from Eq.(\ref{eq:8}), the first inequality is due to the convexity (\ref{eq:the convexity of power}) of $\fP_{C_\bE}$, and the second inequality is obtained because $\fP_{C_\bE}(\Theta_2;\ids{S}) \leq \fP_{C_\bE}(\Theta_2; \Theta_2) = 0$ for $\Theta_2 \in \cO_{\mCMI}^{C_\bE}$.

  \item (Monotonicity) For $\Phi_1, \Phi_2 \in \cO_{\mCMI}^{C_\bE}$, we show that $\fC_{\mCMI}^{C_\bE}(\Phi_1\Theta\Phi_2) \leq \fC_{\mCMI}^{C_\bE}(\Theta)$ for any CPTP map $\Theta$. If $\Theta\in \cO_{\mCMI}^{C_\bE}$, then $\Phi_1\Theta\Phi_2\in \cO_{\mCMI}^{C_\bE}$ because $C_\bE\big(\Phi_1\Theta\Phi_2[\rho]\big)\leq C_\bE\big(\Theta\Phi_2[\rho]\big)\leq C_\bE\big(\Phi_2[\rho]\big)\leq C_\bE(\rho)$ for all $\rho \in \cD_S$. Therefore, we have $\fC_{\mCMI}^{C_\bE}(\Phi_1\Theta\Phi_2) = \fC_{\mCMI}^{C_\bE}(\Theta) = 0$. And, if $\Phi_1\Theta\Phi_2\in\cO_{\mCMI}^{C_\bE}$ for any CPTP map $\Theta$, we can easily see that $0=\fC_{\mCMI}^{C_\bE}(\Phi_1\Theta\Phi_2) \leq \fC_{\mCMI}^{C_\bE}(\Theta)$.
       Finally, we consider the case $\Phi_1\Theta\Phi_2\notin \cO_{\mCMI}^{C_\bE}$  (it is possible only when $\Theta\notin \cO_{\mCMI}^{C_\bE}$). In this case, we have\\
\parbox{6.8cm}{
\begin{eqnarray*}
\fC_{\mCMI}^{C_\bE}(\Phi_1\Theta\Phi_2) &=& \fP_{C_\bE}(\Phi_1\Theta\Phi_2;\ids{S})\\
&=& \max_{\rho\in\cD_S}\big\{C_\bE\big(\Phi_1\Theta\Phi_2[\rho]\big)-C_\bE(\rho)\big\}\\
&\leq& \max_{\rho\in\cD_S}\big\{C_\bE\big(\Phi_1\Theta\Phi_2[\rho]\big)-C_\bE\big(\Phi_2[\rho]\big)\big\}\\
&\leq& \max_{\rho\in\cD_S}\big\{C_\bE\big(\Phi_1\Theta[\rho]\big)-C_\bE(\rho)\big\}\\
&\leq& \max_{\rho\in\cD_S}\big\{C_\bE\big(\Theta[\rho]\big)-C_\bE(\rho)\big\}\\
&=& \fP_{C_\bE}(\Theta;\ids{S}) = \fC_{\mCMI}^{C_\bE}(\Theta),
\end{eqnarray*}}\hfill
\parbox{.3cm}{\begin{eqnarray}\label{eq:the monotonicity of power}\end{eqnarray}}\\
where we use $C_\bE\big(\Phi[\rho]\big) \leq C_\bE(\rho)$ for $\Phi \in \cO_{\mCMI}^{C_\bE}$ in the third and fifth lines, and inequality in the fourth line follows from an increase of the set of states over which we maximize because $\Phi_2[\cD_S] \subset \cD_S$.
\end{enumerate}
\end{proof}


\section{Application of Dynamical Coherence Measure for a POVM with four elements}

The advantages of $\fC_{\mCMI}^{C_\bE}$ are that dynamical resources can be quantified through a relatively simple calculation, as in Eqs. (\ref{eq:8}), and that an extensive application is possible even in the special cases where free states are not given. In this section we illustrate this for certain situations.
Let $\bE = \{\frac{1}{2}\out{\phi_k}{\phi_k}\}_{k=0}^3$ be a POVM on $\cH^{S_0} = \cC^2$ with $\ket{\phi_k} = 1/\sqrt{2}(\ket{0}+\omega^{k}\ket{1})$, where $\omega = \exp(\pi i/2)$, and $\bP = \{\out{\varphi_k}{\varphi_k}\}_{k=0}^3$ be a Naimark extension of it on $\cH^{S}(=\cH^{S_0}\ot\cH^{S_1})=\cC^4$ (here $\cC^4$ is equivalent to $\cC^2\ot\cC^2$ with the identification $\ket{0}^{\cC^4} = \ket{00}^{\cC^2\ot\cC^2}, \ket{1}^{\cC^4} = \ket{01}^{\cC^2\ot\cC^2}, \ket{2}^{\cC^4} = \ket{10}^{\cC^2\ot\cC^2}$ and $\ket{3}^{\cC^4} = \ket{11}^{\cC^2\ot\cC^2}$), where
\bea
\ket{\varphi_0}&=&\frac{1}{2}\big(\ket{0}+\ket{1}+\sqrt{2}\ket{2}\big)\\
\ket{\varphi_1}&=&\frac{1}{2}\big(\ket{0}+i\ket{1}-\exp(\frac{\pi i}{4})\ket{2}+\exp(\frac{\pi i}{4})\ket{3}\big)\\
\ket{\varphi_2}&=&\frac{1}{2}\big(\ket{0}-\ket{1}-\sqrt{2}i\ket{3}\big)\\
\ket{\varphi_3}&=&\frac{1}{2}\big(\ket{0}-i\ket{1}-\exp(-\frac{\pi i}{4})\ket{2}+\exp(\frac{3\pi i}{4})\ket{3}\big).
\eea
Using the formula (\ref{eq:3}), the relative entropy of POVM-based coherence for quantum state $\rho\in \cD_{S_0}$ is
\be
C_{\bE,r}(\rho) = C_{\bP,r}(\Phi[\rho]) = H[\{p_k\}] - S(\rho),
\ee
where $p_k = \frac{\bra{\phi_k}\rho\ket{\phi_k}}{2}.$

\begin{figure}[t]
\centering
\includegraphics[height=.54\textheight]{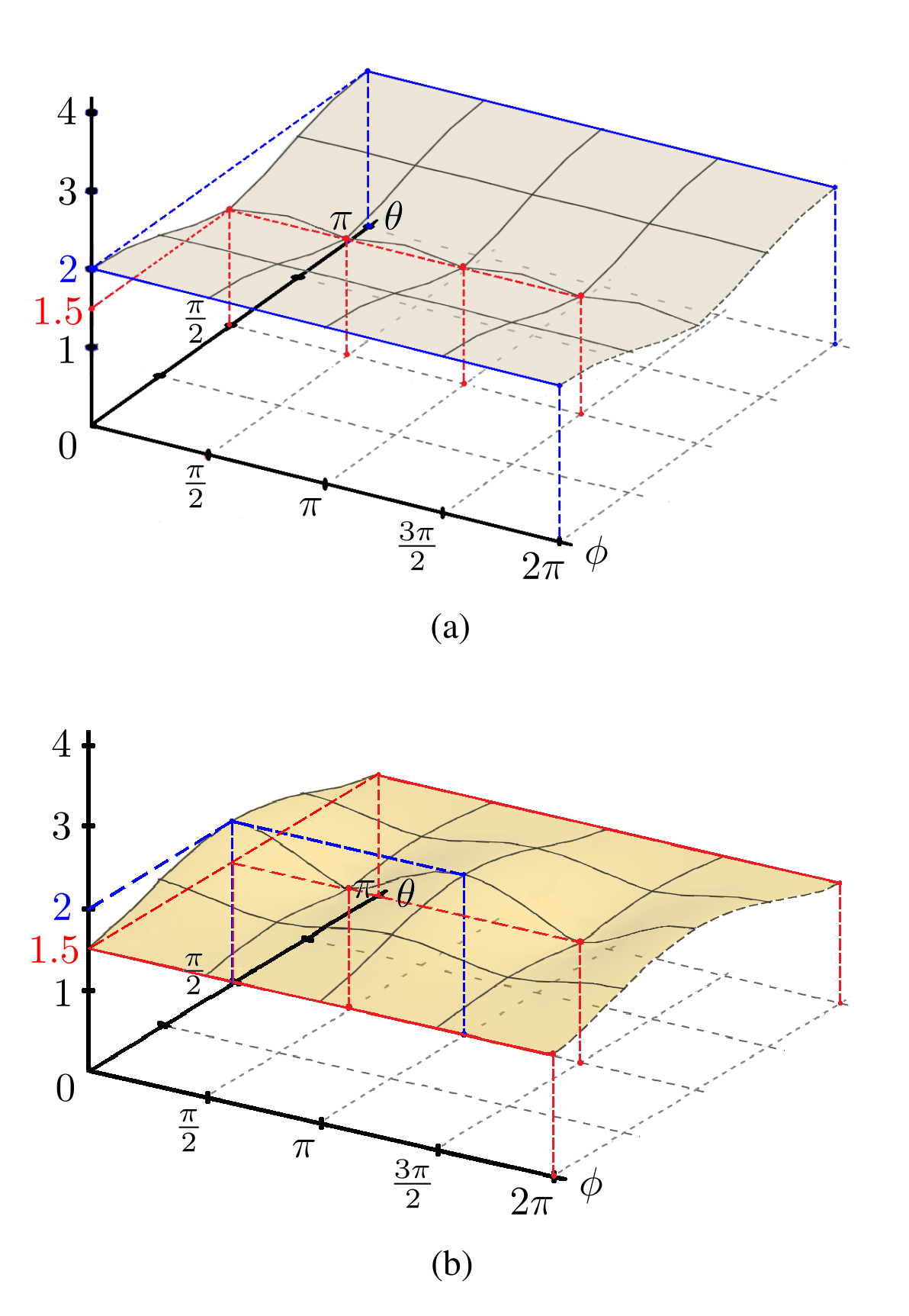}
\caption{\label{fig:1} (a) Graph for $C_{\bE,r}(\out{\psi_{\theta,\phi}}{\psi_{\theta,\phi}})$. Coherence is quantified via $C_{\bE,r}$ for pure states $|\psi_{\theta,\phi}\rangle$. We have a minimum of $1.5$ at $\theta=\frac{\pi}{2}$ and $\phi=0, \frac{\pi}{2}, \pi$ or $\frac{3\pi}{2}$, and a maximum of $2$ at $\theta=0$ or $\pi$.  (b) Graph for $C_{\bE,r}(\Theta_{\textmd{max}}[\out{\psi_{\theta,\phi}}{\psi_{\theta,\phi}}])$. The unitary operation $\Theta_{\textmd{max}}$ causes a maximum increment of $0.5$ in coherence where $\theta=\frac{\pi}{2}$ and $\phi=0$ or $\pi$.}
\end{figure}

We first investigate how the degree of POVM-based coherence is distributed for the pure quantum states.
Consider a state $\ket{\psi_{\theta,\phi}} = \cos{(\frac{\theta}{2})}\ket{0}+\exp(i\phi)\sin{(\frac{\theta}{2})}\ket{1}\ (0\leq\theta\leq \pi,\ 0\leq \phi <2\pi)$ on the Bloch sphere surface. We have
\be
C_{\bE,r}(\out{\psi_{\theta,\phi}}{\psi_{\theta,\phi}}) = H[\{p_k\}],
\ee
where $p_k = |\cos{(\frac{\theta}{2})} + \exp\{k(\frac{\pi i}{2})-i\phi\}\sin{(\frac{\theta}{2})}|^{2}$ for $k=0,1,2,3$ \big(see Fig. \ref{fig:1}(a)\big).
Regardless of $\phi$, when $\theta = 0$ or $\pi$, we have the maximal POVM-based coherence $C_{\bE,r}(\out{\psi_{\theta,\phi}}{\psi_{\theta,\phi}})= 2$ (the maximal POVM-based coherence is obtained in pure states due to the convexity of $C_{\bE,r}$).
On the contrary, when $\theta = \frac{\pi}{2}$ and $\phi= 0,\frac{\pi}{2},\pi$ or $\frac{3\pi}{2}$, we find $C_{\bE,r}(\out{\psi_{\theta,\phi}}{\psi_{\theta,\phi}})= 1.5$ as the minimum value for the pure states.

Now, we consider the degree of POVM-based coherence in the mixed states. For given $\ket{\psi_{\theta,\phi}}$, let there exists another pure state $\ket{\psi'_{\theta,\phi}}$ such that $\iinner{\psi_{\theta,\phi}}{\psi'_{\theta,\phi}}=0$. Then, every quantum state can be represented in the following form:
\be
\rho_{p,\theta,\phi} = p\out{\psi_{\theta,\phi}}{\psi_{\theta,\phi}} + (1-p)\out{\psi'_{\theta,\phi}}{\psi'_{\theta,\phi}},
\ee
with $0\leq p\leq 1, 0\leq \theta \leq \pi$ and $0\leq \phi < 2\pi$.
When $p$ is fixed, we have the following maximum and minimum values of $C_{\bE,r}(\rho_{p,\theta,\phi})$:
\bea
\max_{\theta,\phi}\{C_{\bE,r}(\rho_{p,\theta,\phi})\} = 2-H(p)\quad \ \ \  &&\textmd{when}\ \theta=0\ \textmd{or}\ \pi,\\
\min_{\theta,\phi}\{C_{\bE,r}(\rho_{p,\theta,\phi})\} = 1.5-\frac{H(p)}{2}\quad &&\textmd{when}\ \theta=\frac{\pi}{2}, \phi =0\ \textmd{or}\ \frac{\pi}{2},
\eea
where $H(p) = S(\rho_{p,\theta,\phi}) = -p\log p -(1-p)\log (1-p)$. Thus, we obtain an inequality
\be\label{eq:13}
1.5-\frac{H(p)}{2}\leq C_{\bE,r}(\rho_{p,\theta,\phi})\leq 2-H(p)
\ee
for any quantum state $\rho_{p,\theta,\phi}$ with fixed $p$. Moreover, we obtain the minimal POVM-based coherence $C_{\bE,r}(\rho_{p,\theta,\phi}) = 1$ when $p=\frac{1}{2}$.
We hereby confirm that the POVM-based incoherent state for $\bE$ does not exist in $\cD_S$.

Next, we quantify dynamical resources using the measure $\fC_{\mCMI}^{C_{\bE,r}}$ defined in the preceding section.
We consider increasing the degree of POVM-based coherence via the unitary operations $\Theta$, \emph{i.e.}, $\Theta(\rho) = U\rho U^{\dagger}$ with $U$ a unitary operator. We first note that any unitary operation $\Theta$ is only involved in changing $\theta$ and $\phi$ without changing $p$, \emph{i.e.}, $\Theta(\rho_{p,\theta,\phi}) = \rho_{p,\theta',\phi'} \ (0\leq \theta' \leq \pi, 0\leq \phi' < 2\pi)$. Thus, for any state $\rho_{p,\theta,\phi}$, we have
\bea
C_{\bE,r}\big(\Theta[\rho_{p,\theta,\phi}]\big) - C_{\bE,r}(\rho_{p,\theta,\phi}) &\leq& 2-H(p)- \big(1.5-\frac{H(p)}{2}\big)\\
 &=& 0.5 - \frac{H(p)}{2} \leq 0.5,
\eea
where the first inequality is due to Ineq. (\ref{eq:13}). Therefore, $\fC_{\mCMI}^{C_{\bE,r}}(\Theta) \leq 0.5$ for any unitary operation $\Theta$, and the equality is saturated with the following unitary operator:
\bea
U_{\textmd{max}} = |0\rangle \frac{\langle 0|+\langle 1| }{\sqrt{2}} + |1\rangle \frac{\langle 0|- \langle 1| }{\sqrt{2}}.
\eea

The unitary operator serves to move symmetrically the states on Bloch sphere around the axis from $\cos(\frac{\pi}{8})\ket{0}+\sin(\frac{\pi}{8})\ket{1}$ to $\cos(\frac{3\pi}{8})\ket{0}-\sin(\frac{3\pi}{8})\ket{1}$, and when $\Theta_{\textmd{max}}$ is a unitary operation defined by $U_{\textmd{max}}$, \emph{i.e.}, $\Theta_{\textmd{max}}[\rho] = U_{\textmd{max}}\rho U_{\textmd{max}}^{\dagger}$, the change in POVM-based coherence for the states on Bloch sphere by $\Theta_{\textmd{max}}$ is given by the difference in Fig. \ref{fig:1}(a) and Fig. \ref{fig:1}(b). In particular, when $p=0$ or $1,\theta=\frac{\pi}{2}$ and $\phi=0$ or $\pi$, the increase in the POVM-based coherence due to $\Theta_{\textmd{max}}$ amounts to $0.5$ which means $\fC_{\mCMI}^{C_{\bE,r}}(\Theta_{\textmd{max}}) = 0.5$.
Since the action of the unitary operator $U_{\textmd{max}}$ is to move symmetrically the states on Bloch sphere, the overall mean density of POVM-based coherence remains unchanged, and an increase in the POVM-based coherence can be expected only for the subset of the quantum states.

We also find the incoherent operations induced by measure $C_{\bE,r}$ through the cycle of degree of POVM-based coherence for $\theta$ and $\phi$.
Fig. \ref{fig:1}(a) shows that the unitary operations, which symmetrically move the states on Bloch sphere around some axis, maintain the POVM-based coherence.
Let us consider, for example, a unitary operation $\Lambda_{\textmd{min}}[\rho] = U_{\textmd{min}}\rho U_{\textmd{min}}^{\dagger}$ for $\rho \in \cD_S$ and a unitary operator $U_{\textmd{min}} =\out{0}{1}+\out{1}{0}$.
Then
\bea
U_{\textmd{min}}\ket{\psi_{\theta,\phi}} &=& e^{i\phi}\sin(\frac{\theta}{2})\ket{0}+ \cos(\frac{\theta}{2})\ket{1}\\
&=& \cos(\frac{\pi-\theta}{2})\ket{0}+ e^{i(2\pi-\phi)}\sin(\frac{\pi-\theta}{2})\ket{1},
\eea
and $C_{\bE,r}(\out{\psi_{\pi-\theta,2\pi-\phi}}{\psi_{\pi-\theta,2\pi-\phi}}) = C_{\bE,r}(\out{\psi_{\theta,\phi}}{\psi_{\theta,\phi}})$ from Fig. \ref{fig:1}(a). Hence,
\be
C_{\bE,r}\big(\Lambda_{\textmd{min}}[\out{\psi_{\theta,\phi}}{\psi_{\theta,\phi}}]\big) = C_{\bE,r}(\out{\psi_{\theta,\phi}}{\psi_{\theta,\phi}})
\ee
for $0\leq\theta\leq \pi,\ 0\leq \phi <2\pi$. Above result is also valid for any mixed state obtained by combining such pure states, i.e., $C_{\bE,r}\big(\Lambda_{\textmd{min}}[\rho]\big) = C_{\bE,r}(\rho)$ for $\rho \in \cD_{S}$, because $\Lambda_{\textmd{min}}$ preserves the entropies after the implementation of $\bE$ without changing phase. Therefore, $\fC_{\mCMI}^{C_{\bE,r}}(\Lambda_{\textmd{min}}) = 0$. This means that $\Lambda_{\textmd{min}}$ is an incoherent operation induced by the coherence measure $C_{\bE,r}$ (CMIO by $C_{\bE,r}$). Again, the same result is seen for a unitary operator $U_{\textmd{min}}' = \out{0}{0} - \out{1}{1}$.

Furthermore, consider the mixed unitary operation $\Theta_{\textmd{mixed}}(\rho) = \sum_i p_i \Theta_{i}(\rho)$ for some unitary operations $(\Theta_{i})$ and a probability distribution $(p_i)$. Then $\fC_{\mCMI}^{C_{\bE,r}}(\Theta_{\textmd{mixed}}) \leq \sum_i p_i\fC_{\mCMI}^{C_{\bE,r}}(\Theta_{i}) \leq 0.5$ follows from the convexity of $\fC_{\mCMI}^{C_{\bE,r}}$. Likewise, we find the following CMIO by $C_{\bE,r}$:
\bea
\Lambda_{\textmd{mixed}} = p\Lambda_{\textmd{min}} + (1-p)\Lambda'_{\textmd{min}},
\eea
where $\Lambda'_{\textmd{min}}[\rho] = U'_{\textmd{min}}\rho (U'_{\textmd{min}})^{\dagger}$ for $\rho \in \cD_S$, and $0< p < 1$.

\section{Conclusion and Summary}

We presented an alternate definition of incoherent operations, induced by coherence measures, to overcome the limitations of the traditional framework of POVM-based resource theory. This is not only consistent with the definition of existing incoherent operations in the resource theory of general coherence, but also makes it simpler to determine whether or not operations are incoherent, even if the free states are not given or cannot be determined.
Moreover, we proposed a protocol that allow us to quantify dynamical coherence when no free states are given, based on the newly defined incoherent operation. Finally, as an example, we applied our dynamical coherence theory in the case of POVM-based coherence.

\begin{acknowledgments}
This project is supported by the National Natural Science Foundation of China (Grants No. 12050410232, No. 12031004, and No. 61877054).
\end{acknowledgments}


%

\end{document}